\documentclass[12pt]{article}%

\newcommand{\UsePackage}[1]{%
  \IfFileExists{../styles/#1.sty}{%
      \usepackage{../styles/#1}%
   }{%
      \IfFileExists{./styles/#1.sty}{%
         \usepackage{styles/#1}%
      }{%
         \usepackage{#1}%
      }%
   }%
}

\IfFileExists{sariel_computer.sty}{\def\sarielComp{1}}{}
\ifx\sarielComp\undefined%
\newcommand{\SarielComp}[1]{}
\newcommand{\NotSarielComp}[1]{#1}%
\else
\newcommand{\SarielComp}[1]{#1}%
\newcommand{\NotSarielComp}[1]{}%
\fi
\newcommand{\IfPrinterVer}[2]{#2}%

\usepackage[cm]{fullpage}%
\usepackage{amsmath}%
\usepackage{amssymb}%
\usepackage{xcolor}%

\usepackage[amsmath,thmmarks]{ntheorem}%
\theoremseparator{.}%

\usepackage{titlesec}%
\titlelabel{\thetitle. }%
\usepackage{xcolor}%
\usepackage{mleftright}%
\usepackage{xspace}%
\usepackage{paralist}%
\usepackage{hyperref}%
\usepackage{graphicx}%

\newcommand{\hrefb}[3][black]{\href{#2}{\color{#1}{#3}}}%

\IfPrinterVer{%
   \usepackage{hyperref}%
}{%
   \usepackage{hyperref}%
   \hypersetup{%
      breaklinks,%
      ocgcolorlinks, colorlinks=true,%
      urlcolor=[rgb]{0.25,0.0,0.0},%
      linkcolor=[rgb]{0.5,0.0,0.0},%
      citecolor=[rgb]{0,0.2,0.445},%
      filecolor=[rgb]{0,0,0.4},
      anchorcolor=[rgb]={0.0,0.1,0.2}%
   }
}

\theoremseparator{.}%

\theoremstyle{plain}%
\newtheorem{theorem}{Theorem}[section]

\newtheorem{lemma}[theorem]{Lemma}

\newtheorem{corollary}[theorem]{Corollary}

\theoremstyle{plain}%
\theoremheaderfont{\sf} \theorembodyfont{\upshape}%
\newtheorem*{remark:unnumbered}[theorem]{Remark}%
\newcommand{\myqedsymbol}{\rule{2mm}{2mm}}

\theoremheaderfont{\em}%
\theorembodyfont{\upshape}%
\theoremstyle{nonumberplain}%
\theoremseparator{}%
\theoremsymbol{\myqedsymbol}%
\newtheorem{proof}{Proof:}%

\newcommand{\atgen}{\symbol{'100}}
\newcommand{\SarielThanks}[1]{\thanks{Department of Computer Science;
      University of Illinois; 201 N. Goodwin Avenue; Urbana, IL,
      61801, USA; {\tt sariel\atgen{}illinois.edu}; {\tt
         \url{http://sarielhp.org/}.} #1}}

\newcommand{\RahulThanks}[1]{\thanks{Department of Computer Science;
      University of Illinois; 201 N. Goodwin Avenue; Urbana, IL,
      61801, USA; {\tt saladi.rahul\atgen{}gmail.com.} #1}}

\newcommand{\HLink}[2]{\hyperref[#2]{#1~\ref*{#2}}}
\newcommand{\HLinkSuffix}[3]{\hyperref[#2]{#1\ref*{#2}{#3}}}

\newcommand{\thmlab}[1]{{\label{theo:#1}}}
\newcommand{\thmref}[1]{\HLink{Theorem}{theo:#1}}

\newcommand{\lemlab}[1]{\label{lemma:#1}}
\newcommand{\lemref}[1]{\HLink{Lemma}{lemma:#1}}%

\newcommand{\itemlab}[1]{\label{item:#1}}
\newcommand{\itemref}[1]{\HLinkSuffix{}{item:#1}{}}

\providecommand{\eqlab}[1]{}%
\renewcommand{\eqlab}[1]{\label{equation:#1}}

\newcommand{\remove}[1]{}%
\newcommand{\Set}[2]{\left\{ #1 \;\middle\vert\; #2 \right\}}
\newcommand{\pth}[2][\!]{\mleft({#2}\mright)}%
\newcommand{\pbrcx}[1]{\left[ {#1} \right]}%
\newcommand{\Ex}[2][\!]{\mathop{\mathbf{E}}#1\pbrcx{#2}}

\newcommand{\ceil}[1]{\left\lceil {#1} \right\rceil}
\newcommand{\floor}[1]{\left\lfloor {#1} \right\rfloor}

\newcommand{\brc}[1]{\left\{ {#1} \right\}}
\newcommand{\cardin}[1]{\left| {#1} \right|}%

\renewcommand{\th}{th\xspace}
\usepackage[inline]{enumitem}

\newlist{compactenumA}{enumerate}{5}%
\setlist[compactenumA]{topsep=0pt,itemsep=-1ex,partopsep=1ex,parsep=1ex,%
   label=(\Alph*)}%

\newlist{compactenuma}{enumerate}{5}%
\setlist[compactenuma]{topsep=0pt,itemsep=-1ex,partopsep=1ex,parsep=1ex,%
   label=(\alph*)}%

\newlist{compactenumI}{enumerate}{5}%
\setlist[compactenumI]{topsep=0pt,itemsep=-1ex,partopsep=1ex,parsep=1ex,%
   label=(\Roman*)}%

\newlist{compactenumi}{enumerate}{5}%
\setlist[compactenumi]{topsep=0pt,itemsep=-1ex,partopsep=1ex,parsep=1ex,%
   label=(\roman*)}%

\providecommand{\Mh}[1]{#1}%

\newcommand{\URad}{\Mh{\Delta}}%
\newcommand{\sRad}{\Mh{\delta}}%

\newcommand{\Event}{\Mh{\mathcal{E}}}%
\newcommand{\ProbLTR}{\mathbb{P}}%

\newcommand{\Prob}[1]{\mathop{\ProbLTR} \mleft[ #1 \mright]}%
\newcommand{\ProbCond}[2]{\mathop{\ProbLTR}\!\left[%
       #1 \;\middle\vert\; #2 \right]}

\usepackage{stmaryrd}%

\newcommand{\IRY}[2]{\left\llbracket #1:#2 \right\rrbracket}
\newcommand{\eps}{\varepsilon}%

\newcommand{\etal}{\textit{et~al.}\xspace}
\newcommand{\pathY}[2]{\Mh{\mathrm{p}}\pth{#1, #2}}%
\newcommand{\Gyori}{Gy{\" o}ri\xspace}%
\newcommand{\G}{\Mh{G}}%
\newcommand{\VV}{\Mh{V}}%
\newcommand{\EE}{\Mh{E}}%
\newcommand{\Partition}{\Mh{{P}}}%
\newcommand{\BallY}[2]{\Ball\pth{#1, #2}} \newcommand{\Ball}{\Mh{b}}%

\definecolor{blue25emph}{rgb}{0, 0, 11}
\newcommand{\emphic}[2]{%
   \textcolor{blue25emph}{%
      \textbf{\emph{#1}}}%
   \index{#2}}

\newcommand{\emphi}[1]{\emphic{#1}{#1}}

\SarielComp{\usepackage{sariel_colors}}%

\providecommand{\BibLatexMode}[1]{}
\providecommand{\BibTexMode}[1]{#1}

\ifx\UseBibLatex\undefined%
  \renewcommand{\BibLatexMode}[1]{}
  \renewcommand{\BibTexMode}[1]{#1}
\else
  \renewcommand{\BibLatexMode}[1]{#1}
  \renewcommand{\BibTexMode}[1]{}
\fi

\BibLatexMode{%
   \usepackage[bibencoding=utf8,style=alphabetic,backend=biber,
   sortlocale=en_US]{biblatex}%
   \usepackage[english]{babel}
   \usepackage{csquotes}
   \UsePackage{sariel_biblatex}%
}

\BibLatexMode{%
}

\begin{document}

\title{Two (Known) Results About Graphs with No Short Odd Cycles}

\author{%
   Sariel Har-Peled\SarielThanks{Work on this paper
      was partially supported by a NSF AF awards
      CCF-1421231, and 
      CCF-1217462.  
   }%
   \and%
   Saladi Rahul\RahulThanks{}  %
}

\date{\today}

\maketitle

\begin{abstract}
    Consider a graph with $n$ vertices where the shortest odd cycle is
    of length $>2k+1$.  We revisit two known results about such
    graphs: %
    \medskip
    \begin{compactenumA}
        \item Such a graph is almost bipartite, in the sense that it
        can be made bipartite by removing from it
        $O\bigl( (n/k) \log (n/k) \bigr)$ vertices. While this result
        is known \cite{gkl-gwsoc-97} -- our new proof seems to yield
        slightly better constants, and is (arguably) conceptually
        simpler.

        To this end, we state (and prove) a version of C{K}R
        partitions \cite{ckr-aazep-04, frt-tbaam-04} that has a small
        vertex separator, and it might be of independent interest.
        While this must be known in the literature, we were unable to
        find a reference to it, and it is included for the sake of
        completeness.
        
        \medskip
        \item While such graphs can be quite dense (e.g., consider a
        the bipartite clique, which has no odd cycles), they have a
        large independent set. Specifically, we prove that such graphs
        have independent sets of size
        $\geq \bigl(1-o(1)\bigr)n^{k/(k+1)}$. Again, this result is
        known and is implied by the work of Shearer \cite{s-indgl-95},
        but our proof is simpler and (seems to) yield a better
        constant.
    \end{compactenumA}
\end{abstract}




\section{Graphs with no short odd cycles are almost %
   bipartite}

We start by proving a variant of a result of Fakcharoenphol \etal
\cite{frt-tbaam-04} about C{K}R partitions \cite{ckr-aazep-04}.
Fakcharoenphol \etal proved a bound on the probability of an edge to
be cut by the random partitions of C{K}R (and thus the number of such
edges being cut), while we are interested (somewhat imprecisely) in
the number of boundary vertices.  In particular, we show that a graph
with $n$ vertices can be broken into disconnected clusters of diameter
at most $\URad$, by removing only $O( (n/\URad) \log (n/\URad))$
vertices.  The proof is slightly simpler, and the constants are
slightly better than Fakcharoenphol \etal \cite{frt-tbaam-04}, and is
included for the sake of completeness -- while we were unable to find
a reference for it in the literature, we assume this result is
known. See \thmref{frt:vertices} for the precise statement.

We then use this to prove that graphs that have only long odd cycles
can be converted into bipartite graphs by removing ``few'' vertices.
See \thmref{remove} below for the precise statement.


\subsection{Partitions with small separating set}

Let $\G = (\VV,\EE)$ be an unweighted and undirected graph with $n$
vertices.  A \emphi{partition} of $\G$ is a set
$\Partition = \brc{C_1,\ldots,C_m}$ of disjoint subsets of $\VV$ such
that $\bigcup_{C_i \in \Partition} C_i = \VV$. The sets $C_i$ are
\emphi{clusters}.  For $x\in \VV$ and a partition $\Partition$, let
$\Partition(x)$ denote the unique cluster of $\Partition$ containing
$x$.

\newcommand{\distGY}[2]{\Mh{d}\pth{#1,#2}}%

In the following, for a vertex $v \in \VV$, and an integer $R$, let
$\BallY{v}{R} = \Set{x \in \VV}{\distGY{v}{x} \leq R}$ be the
\emphi{ball} of radius $r$ centered at $v$, where $\distGY{v}{x}$ is
the shortest path distance in $\G$ between $v$ and $x$ (i.e., it is
the minimum number of edges on a path between $v$ and $x$ in $\G$).

\paragraph{Constructing the partition.}

Let $\URad = 4\sRad$ be a prescribed parameter, for some integer
$\sRad$.  Choose, uniformly at random, a permutation $\pi$ of $\VV$
and a value
\begin{equation*}
    R \in \IRY{\URad/4}{\URad/2}%
    =%
    \brc{ \URad/4, \URad/4 +1, \ldots, \URad/2} .
\end{equation*}
The partition is now defined as follows: A vertex $x\in \VV$ is
assigned to the cluster $C_y$ of $y$, where $y$ is the first vertex in
the permutation within distance $\leq R$ from $x$. Formally,
\begin{equation*}
    C_y%
    =%
    \Set{ x \in \VV }{%
       \bigl.
       x \in \BallY{y}{R} \textrm { and } \pi(y)
       \leq \pi(z) \textrm{ for all $z \in \VV$ with } x \in
       \BallY{z}{R}
    },
\end{equation*}
Let $\Partition = \{C_y\}_{y \in \G}$ denote the resulting partition.
In words, once we fix the radius of the clusters $R$, we start
scooping out balls of radius $R$ centered at the vertices of the
random permutation $\pi$. At the $i$\th stage, we scoop out only the
remaining mass at the ball centered at $x_i$ of radius $R$, where
$x_i$ is the $i$\th vertex in the random permutation.

\newcommand{\guardsX}[1]{\Mh{\mathrm{guards}\pth{#1}}}

\subsubsection{Properties}

An edge $uv \in \EE$ is a \emphi{cross edge} of a partition
$\Partition$ if $u$ and $v$ belong to different clusters of
$\Partition$. A vertex $u$ of a partition $\Partition$ computed above
is a \emphi{guard} if $u \in C_y \in \Partition$, and
$\distGY{y}{u} = R$. Let $\guardsX{\Partition}$ be the set of all the
vertices of $\VV$ that are guards.

\newcommand{\diameterX}[1]{\Mh{\mathrm{diam}}\pth{#1}}%

\begin{lemma}
    \lemlab{cross}%
    Let $\Partition$ be a partition computed by the above scheme. We
    have the following:
    \begin{compactenumA}
        \item For any $C \in \Partition$, we have
        $\diameterX{C} \leq \URad$, where $\diameterX{C}$ is the
        \emph{diameter} of $C$.

        \item
        Let $uv$ be a cross edge of $\Partition$. Then $u$ is a guard or
        $v$ is a guard.
    \end{compactenumA}
\end{lemma}
\begin{proof}
    (A) A cluster is a subset of a ball of radius $R$, and as such its
    diameter is at most $2R \leq \URad$.

    (B) Assume $u \in C_x$ and $v \in C_y$, and $x$ was before $y$ in
    the permutation (the other case is handled symmetrically). If $u$
    is not a guard, then $\distGY{x}{u} <R$. But then
    $\distGY{x}{v} \leq \distGY{x}{u} +1 \leq R$, and $v$ would be in
    $C_x$. A contradiction.
\end{proof}

\begin{lemma}
    \lemlab{u:F:R:T}%
    Let $\G=(\VV,\EE)$ be an undirected and unweighted graph over $n$
    vertices, and let $\URad = 4\sRad$ be a prescribed parameter (for
    some integer $\sRad$), and let $\Partition$ be the random
    partition of $\G$ generated by the above scheme.  Then, for any
    vertex $x \in \VV$, we have that
    \begin{math}
        \displaystyle \Prob{ \bigl. x \in \guardsX{\Partition}}%
        \leq%
        \frac{4}{\URad} \ln \frac{ \cardin{\BallY{x}{\URad / 2}}}
        {\cardin{\BallY{x}{\URad/4 -1}}}.
    \end{math}
\end{lemma}
\begin{proof}
    Let $U= \BallY{x}{\URad/2}$, $M = \cardin{U}$, and
    $m = \cardin{\BallY{x}{\URad/4 -1}}$.  Arrange the vertices of
    $U$ in increasing distance from $x$, and let $w_1, \ldots, w_{M}$
    denote the resulting order.  Let $\Event_k$ for the event that
    $w_k$ is the \emph{first} vertex in $\pi$ such that
    $x \in C_{w_k}$ and $\distGY{w_k}{x} = R$. Observe that if
    $x \in \guardsX{\Partition}$, then one of the events
    $\Event_1, \ldots, \Event_{M}$ must occur.
    
    Let $\BallY{x}{\URad/4 - 1} = \brc{w_1, \ldots, w_m}$.  Note that
    if $w_k \in \BallY{x}{\URad/4 - 1}$, then $\Prob{\Event_k} = 0$
    since $R > \URad/4 -1$, which implies $\distGY{w_k}{x} < R$.  As
    such, for $i=1, \ldots,m$, we have
    $\Prob{ \Event_1} = \cdots = \Prob{ \Event_m} = 0$.

    Observe that
    \begin{math}
        \Bigl.\Prob{ \bigl. \Event_k }%
        =%
        \Prob{\bigl. \Event_k \cap (R = \distGY{w_k}{x}) }%
        =%
        \Prob{\bigl.R = \distGY{w_k}{x}} \cdot \ProbCond{
           \bigl. \Event_k}{R = \distGY{w_k}{x}}.
    \end{math}
    Since $R$ is uniformly distributed in the interval
    $\IRY{\URad/4}{\URad/2}$, we have that
    \begin{math}
        \Prob{R = \distGY{w_k}{x}}%
        \leq%
        1/(\URad/4 +1)%
        \leq%
        4 /\URad.
    \end{math}

    To bound $\beta_k = \ProbCond{ \Event_k}{R = \distGY{w_k}{x}}$, we
    observe that $w_1, \ldots, w_{k-1}$ are closer (or of the same
    distance) to $x$ than $w_k$. Thus, if any of them appear before
    $w_k$ in $\pi$, then $\Event_k$ does not happen. Thus, $\beta_k$
    is bounded by the probability that $w_k$ is the first to appear in
    $\pi$ out of $w_1, \ldots, w_k$. This probability is $1/k$, and
    thus $\beta_k \leq 1/k$.  We have that
    \begin{equation*}
        \Prob{\bigl.\Event_k}%
        =%
        \Prob{\Bigl. R = \distGY{w_k}{x}} 
        \ProbCond{ \Bigl. \Event_k}{R = \distGY{w_k}{x}}
        \leq%
        \frac{4}{\URad} \cdot \frac{1}{k}.
    \end{equation*}
    We conclude that 
    \begin{math}
        \displaystyle%
        \Prob{\bigl. x \in \guardsX{\Partition}} %
        =%
        \sum_{k=1}^{M} \Prob{\bigl.\Event_k}%
        =%
        \sum_{k=m+1}^{M} \Prob{\bigl. \Event_k} %
        \leq%
        \frac{4}{\URad} \sum_{k=m+1}^{M} \frac{1}{k} \leq%
        \frac{4}{\URad} \ln \frac{M}{m}.
    \end{math}
\end{proof}

We thus get the following result, which is variant of the result of
Fakcharoenphol \etal \cite{frt-tbaam-04}.
\begin{theorem}
    \thmlab{frt:vertices}%
    Let $\G = (\VV,\EE)$ be an unweighted and undirected graph with
    $n$ vertices. Given a parameter $\URad = 4\sRad$, for some
    positive integer $\sRad$, one can randomly partition the graph
    into clusters $C_1, \ldots C_t$, such that:
    \begin{compactenumA}
        \item Every cluster has a center vertex $c_i$, such that for
        all $x \in C_i$, we have $\distGY{c_i}{x} \leq \URad/2$ (and
        thus, $\diameterX{C_i} \leq \URad$).

        \item There is a set of vertices $X$, of expected size
        $\leq \frac{4n}{\URad} \ln \frac{4n}{\URad}$, such that
        there is no edge between a vertex of $C_i \setminus X$ and
        $C_j \setminus X$, for all $i\neq j$.
    \end{compactenumA}
\end{theorem}
\begin{proof}
    The construction is described above.  We take
    $X = \guardsX{\Partition}$, which by \lemref{cross}, is the
    desired separating set.  In a connected graph, for any vertex
    $x \in \VV$, we have $ \cardin{\BallY{x}{\URad / 2}} \leq n$ and
    $\cardin{\BallY{x}{\URad/4 -1}} \geq \URad/4$.  As such, by
    linearity of expectation and \lemref{u:F:R:T}, we have
    \begin{equation*}
        \Ex{\bigl.\cardin{\guardsX{\Partition}}}%
        =%
        \sum_{v \in \VV} \Prob{\bigl. x \in \guardsX{\Partition}}%
        \leq%
        n \cdot
        \frac{4}{\URad} \ln \frac{ \cardin{\BallY{x}{\URad / 2}}}
        {\cardin{\BallY{x}{\URad/4 -1}}}.
        \leq%
        \frac{4n}{\URad} \ln \frac{4n}{\URad}.
    \end{equation*}
\end{proof}

\subsection{Application: Graphs with no short odd %
   cycles are almost bipartite}

    Odd cycles have an interesting hereditary property -- if there is
    an odd cycle $C$ that is not simple (or it has a chord), then
    there must be a shorter simple odd cycle in the graph: \smallskip%

\noindent%
\begin{minipage}{0.6\linewidth}
\begin{compactenumI}
    \item \itemlab{I} If $C$ repeats a vertex twice, then it can be decomposed
    into two shorter cycles.  One of these cycles must be odd.

    \item If $C$ repeats an edge twice, then it repeats a vertex
    twice, and \itemref{I} applies.

    \item If there is an edge between two non-adjacent vertices of
    $C$, then it can be split in a similar fashion.
\end{compactenumI}
\smallskip%
\end{minipage}%
\hfill%
\begin{minipage}{0.37\linewidth}
    \includegraphics[page=1,width=0.49\linewidth]{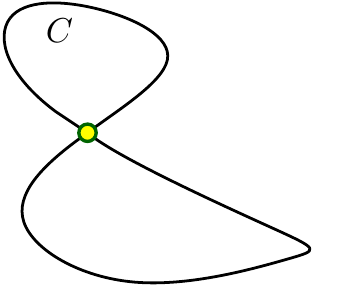}
    \includegraphics[page=2,width=0.49\linewidth]{figs/odd_split}
\end{minipage}%

\noindent%
Applying this argument repeatedly results in a simple odd cycle that
is a sub-cycle of $C$.  Note, that this hereditary property does not
hold for even length cycles.

In particular, using the above partition theorem, we get the following
result of \Gyori \etal \cite{gkl-gwsoc-97} -- our new proof seems to
yield slightly better constants, and is conceptually simpler if
\thmref{frt:vertices} is a given.

\begin{theorem}
    \thmlab{remove}%
    Let $\G = (\VV,\EE)$ be an undirected graph over $n$ vertices,
    such that all the odd cycles in $\G$ are of length $> 2k+1$, for
    some integer $k$.  Then, there is a set $X$ of at most
    $\frac{ n}{\floor{k/2}} \ln \frac{ n } {\floor{k/2}}$ vertices,
    such that removing $X$ from $\G$ results in a bipartite graph.
\end{theorem}
\begin{proof}
    Assume $\G$ is connected, as otherwise we apply the argument below
    to each connected component.
    
    Let $\sRad = \floor{k/2}$, $r = 2\sRad$, and let $\URad =
    4\sRad$. Compute the partition $\Partition$ of
    \thmref{frt:vertices}. This breaks the graph into $t$ clusters
    $C_1, \ldots, C_t$, where all the vertices of $C_i$ are in
    distance at most $r$ from its center $c_i$.

    Observe that there are no two vertices $x,y \in C_i$, such that
    $t = \distGY{c_i}{x} = \distGY{c_i}{y}$ and $xy$ is an edge in the
    graph. If there is such an edge, then consider the cycle
    $\sigma = \pathY{c_i}{x} + xy + \pathY{y}{c_i}$, where
    $\pathY{u}{v}$ denotes the shortest path in $\G$ between $u$ and
    $v$. The cycle $\sigma$ is an odd cycle in $\G$ of length
    \begin{math}
        \distGY{c_i}{x} +1 + \distGY{y}{c_i}%
        =%
        2t + 1%
        \leq%
        2r + 1%
        =%
        4\!\floor{k/2} +1%
        \leq%
        2k +1,
    \end{math}
    which is a contradiction.
    
    An odd cycle can not be fully contained inside a cluster $C_i$,
    since it must contain an edge with both endpoints having the same
    distance to $c_i$, which is by the above impossible. Thus, all odd
    cycles in $\G$ must use a cross edge of $\Partition$. In
    particular, removing the set of vertices
    $X = \guardsX{\Partition}$ from $\G$, implies by \lemref{cross},
    that the remaining graph $\G' = \G \setminus X$ has no cross
    edges, and thus no odd cycles. Namely, $\G'$ is bipartite.
    \thmref{frt:vertices} implies that
    \begin{math}
        \Ex{\bigl.\cardin{X}}%
        \leq%
        \frac{4n}{\URad} \ln \frac{4n}{\URad}%
        =%
        \frac{4n}{4\floor{k/2}} \ln \frac{4n}{4\floor{k/2}}.
    \end{math}
\end{proof}

\begin{corollary}
    Let $\G = (\VV,\EE)$ be an undirected graph with $n$ vertices,
    such that all the odd cycles in $\G$ are of length $> \eps n$.
    Then one can remove $O(\eps^{-1} \log \eps^{-1})$ vertices from
    $\G$ and make the graph bipartite. More precisely, if
    $\eps n \geq 40$, then one can remove $\leq (5/\eps) \ln(5/\eps)$
    vertices and make the graph bipartite.
\end{corollary}
\begin{proof}
    Let $k = \floor{(\eps n -1)/2}$. By \thmref{remove}, there is a
    set $X$ of size
    \begin{equation*}
        \frac{ n}{\floor{k/2}} \ln \frac{ n } {\floor{k/2}}        
        \leq%
        \frac{ n}{\eps n/4 -2} \ln \frac{ n } {\eps n/4 -2}        
        \leq 
        \frac{ 5}{\eps} \ln \frac{ 5 } {\eps},
    \end{equation*}
    since
    \begin{math}
        \eps n  \geq 40,
    \end{math}
    such that removing $X$ from $\G$ results in a bipartite graph.
\end{proof}


\newcommand{\ISetY}[2]{\Mh{\mathcal{I}}_{#1}\pth{#2}}%
\newcommand{\NbrY}[2]{\Mh{\mathcal{N}}_{#1}\pth{#2}}%
\newcommand{\BFS}{\textsf{BFS}\xspace}%

\section{Graphs with no short odd cycles have large %
   independent sets}

Interestingly, graphs with no short odd cycles have large independent
sets, and these sets can be computed by a simple algorithm. The
starting point is the greedy algorithm for independent set in a graph,
which works by finding a vertex of low degree, removing itself and its
neighbors, and repeating this till the graph is exhausted. The idea is
to do better by considering not only the direct neighbors of a vertex,
but rather inspecting the neighborhood of the vertex till a certain
depth. The layers of the neighborhood of a vertex must grow quickly
for each one of them not to be a good independent set to harvest,
implying that sooner or later a good layer would be encountered.

The following result is known and is implied by the work of Shearer
\cite{s-indgl-95}, but our version is simpler and (seems to) yield a
better constant.

\begin{lemma}
    Let $k > 1$ be an integer parameter, and let $\G=(\VV,\EE)$ be an
    undirected graph with $n$ vertices and $m$ edges, with all odd
    cycles in $\G$ being of length $>2k+1$, for some $k\geq 1$. Then,
    $\G$ has an independent set of size $(1 - o(1))n^{k/(k+1)}$, which
    can be computed in $O(n + m)$ time.
\end{lemma}
\begin{proof}
    The algorithm repeatedly deletes vertices from the graph $\G$ and
    add some of them to the computed independent set.  For any vertex
    $v$ in the current graph $\G$, let $L_{i}(v)$ be the set of
    vertices of $\G$ at distance exactly $i$ from $v$ (as such,
    $L_{-1}(v) = \emptyset$ and $L_0(v) = \brc{v}$).  Let
    $d_i(v) = \cardin{L_i(v)}$ -- in particular, $d_0(v) = 1$, and
    $d_1(v) = d(v)$ is the degree of $v$.  Let $X$ be a set of
    independent vertices that is initially empty.

    Set $K = \ceil{n^{1/(k+1)}}$.  The algorithm repeatedly picks a
    vertex $v$, and computes a \BFS tree from $v$ level by level. For
    $j > 0$, after computing the $j$\th level of the \BFS tree (i.e,
    $L_j(v)$), the algorithm checks if $d_j(v) \leq K d_{j-1}(v)$. If
    so, the algorithm adds $L_{j-1}(v)$ to the independent set $X$,
    and removes the vertices of
    $L_{j}(v) \cup L_{j-1}(v) \cup \cdots \cup L_0(v)$ from $\G$. The
    algorithm does this till the graph is exhausted.

    Observe, that if the algorithm harvested the set $L_{j-1}(v)$ to
    the independent set, then 
    \begin{equation*}
        d_{j}(v)\leq K d_{j-1}(v)%
        \qquad\text{and}\qquad%
        d_{\ell}(v) > K d_{\ell-1}(v) \qquad\forall \ell=0,\ldots, j-1.
    \end{equation*}
    This implies that 
    \begin{math}
        d_{j-1}(v) > K^{j-1}.
    \end{math}
    As such, if $L_{k+1}(v)$ is being added by the algorithm to the
    output set (i.e., $j=k+2$), this would imply that
    $d_{k+1}(v) > K^{k+1}  \geq n$, which is impossible.
    
    As such, the sets being added to the output set, are of the form
    $L_i(v)$, for $i\leq k$. Such a set is independent, as otherwise,
    there would be an edge $xy \in \EE$ between two vertices of
    $L_i(v)$, but that would imply an odd cycle in $\G$ of length
    $\leq 2k+1$. This readily implies that the computed set is indeed
    an independent set in the original graph.

    As for the size of the independent set computed, observe that when
    $L_{j-1}(v)$ is being added to the output, the number of vertices
    being deleted is of size
    \begin{equation*}
        d_j(v) + d_{j-1}(v) + d_{j-2}(v) + \cdots +d_0(v)
        \leq%
        \pth{K + 1 + \frac{1}{K} +\frac{1}{K^2} + \ldots} d_{j-1}(v)%
        \leq %
        (K+2) d_{j-1}(v)
    \end{equation*}
    Namely, each vertex in the output independent set, pays for at
    most $K+2$ vertices in the graph, and the computed independent set
    is of size $\geq n /(K+2) = c n^{k/(k+1)}$, where
    \begin{equation*}
        c = \frac{n^{1/(k+1)}}{K+2}
        \geq%
        \frac{n^{1/(k+1)}}{n^{1/(k+1)} + 3}
        =%
        1 - \frac{3}{ n^{1/(k+1)} + 3}
        =%
        1 - o(1).
    \end{equation*}

    As for the running time of the algorithm, note that the \BFS
    computation, implemented carefully, takes time proportional to the
    number of edges and vertices being removed. Thus the overall
    running time of the algorithm is linear.
\end{proof}

Bock \etal \cite{bfmr-ssspt-14} claim an $O(n^{2.5})$ time algorithm
to find an independent set of size $\geq n^{k/(k+1)}/3$ for graphs
with all odd cycles in $\G$ being of length $>2k+1$. In comparison,
the above improves both the running time and the approximation
quality.


\BibTexMode{%
 \providecommand{\CNFFSTTCS}{\CNFX{FSTTCS}}  \providecommand{\CNFX}[1]{
  {\em{\textrm{(#1)}}}}

}
\BibLatexMode{\printbibliography}

\end{document}